\documentclass[twocolumn,secnumarabic,amssymb,nobibnotes,aps,prd]{revtex4-1}
\usepackage{graphicx}
\usepackage{bm}
\usepackage{times}
\usepackage{amssymb}
\usepackage{amsthm}
\usepackage{mathrsfs}
\usepackage{pifont}
\usepackage{amsmath}
\usepackage[colorlinks, citecolor=blue]{hyperref}
\usepackage[height=9.7in, width=7.3in, voffset=0.25in]{geometry}
\newtheorem{Definition}{Definition}
\newtheorem{Hypothesis}{Hypothesis}

\begin{document}

\title{Construction of Spacetime Field and Quantization of General Relativity}%

\author{Zhang Chang-kai}%
\email[E-mail:\;]{2007909kxjd@163.com}
\affiliation{Sanming No.2 Middle School, Sanming, Fujian, China}

\begin{abstract}
\noindent
\textbf{Abstract}\ \;\;Theories based on General Relativity or Quantum Mechanics have taken a leading position in macroscopic and microscopic Physics, but fail when used in the other extremity. Thus, we try to establish a new structure of united theory based on General Relativity by forming certain spacetime property and a new model of particle. This theory transforms the Riemann curvature tensor into spacetime density scalar so that gravitational field can be added to the Quantum Mechanics, and supposes the electromagnetic field in General Relativity to be a kind of spacetime fluid. Also the theory carries on Einstein-Cartan Theory about spacetime torsion to form the energy exchange in order that Hubble's law and dark mass can be simply explained. On the basis of the definition of signal, measurement and observation, the theory is able to construct a modified Quantum Equation in curved spacetime compatible with General Relativity, and hence unite the General Relativity and Quantum Mechanics in some range and puts forward a new idea of unification.


\end{abstract}

\maketitle

\section{Introduction}
The contradiction between the General Theory of Relativity and Quantum Mechanics was discovered not long after they were developed into a successful theory. Hence, searching for a united theory has become one of the mainstreams of the Theoretical Physics. Nowadays, there are three main approaches to the united theory: the String Theory, the Deformation Quantization and Loop Quantum Gravity.

The String Theory is a popular candidate of the united theory\cite{STMvol}\cite{STMNS}. This theory supposes the basic element of matter are not the point particles but vibrating strings in a high dimensional space. The newly developed M-Theory supposes that every basic particle and the basic force fields can be united in 11 dimensional space. However, this theory meets difficulties in mathematical tools and experimental verification.

The Loop Quantum Gravity is also a very competitive alternative\cite{LQGQGR}\cite{QG}\cite{LQGOV}. It is the most successful non-perturbative quantization theory and do not need any extra dimensions. Nevertheless, it only touches upon the quantization of gravity and therefore cannot form a complete united theory.

The Deformation Quantization suggests that any nontrivial associative deformation of an algebra of functions can be explained as a kind of quantization. Based on Poisson Geometry, Group Theory and noncommutative geometry, it mainly develops various of quantization in linear and nonlinear field theories. \cite{DQPM}\cite{QM}

And there are also efforts on some other threads toward a united theory such as the development of a quantum field theory in curved spacetime\cite{QFTOPE}\cite{AQFT}. However, these theories have not yet led to a satisfying united theory, either.

We try to construct a new form of a possible united theory (General United Spacetime Field Theory, GUSFT) with the spacetime field and quantized matter distribution. In this theory, the state of a particle system is described by quantized matter distribution, the interaction is described by spacetime field and the evolution is determined as a result.

The present paper is organized as follows. Section \uppercase\expandafter{\romannumeral2} mainly concerns the construction of spacetime field and Section \uppercase\expandafter{\romannumeral3} focuses on the attempt to quantize the General Relativity. In Section \uppercase\expandafter{\romannumeral2}.1, we largely present the definition and formulation of spacetime property, and gain the spacetime formula in the condition of Schwarzschild Solution. In Section \uppercase\expandafter{\romannumeral2}.2, we create the 4-velocity of electromagnetic field and therefore reformulate the energy-momentum tensor. In Section \uppercase\expandafter{\romannumeral2}.3, we discuss the situations of positive and negative energy exchange based on spacetime torsion, and then illustrate its effect on the evolution of universe. In Section \uppercase\expandafter{\romannumeral3}, the definitions of signal, measurement and observation are put forward, on the basis of which the curved Quantum Equation is constructed. And hence, we are able to give the possibility explanation a new meaning. Thus far, we complete the theoretical foundation of GUSFT.

\section{Spacetime Property}
This theory is based on the two basic quantities, spacetime density and spacetime fluid, which aims to describe the gravitational field and the electromagnetic field. We try to consider the gravitational field as a kind of density scalar such that it can be used to indicate the curvature of spacetime. And we consider the electromagnetic field as a kind of perfect fluid, satisfying the energy-momentum tensor. Then, through the electromagnetic field can it be combined with the Quantum Mechanics.

All the tensors in this article are written by abstract index notation. All the inner products of Dirac notation do not include the integral. 

\subsection{Spacetime Density}
According to the General Relativity, the gravitational field is described as the curvature of the spacetime. In Mathematics, the curvature is explained as the change of spacetime geometry. But in Physics, it's natural to find a physical quantity homologous to it. Due to the effort of the attempt to reflect the effect of dark matter and dark energy, it's suitable to suppose this quantity is density. In Einstein's field equation, Riemann curvature tensor takes the leading position. The Riemann tensor represents the curvature degree of spacetime, one of whose significant property is the curve dependence in the parallel-transport of vectors. So we hope it's possible to express this property by the density.

First we consider the parallel-transport of vectors along a closed path. In consideration of the equivalence principle, it's necessary to consider the torsion tensor to be infinitely small compared with the curvature tensor. Suppose there is a $2$-dimensional submanifold $S$ in $n$-dimensional manifold $M$, which has the curves without self-intersection parameterized with $t$ and $s$ and can form a closed path with tangent vector $T^a$ and $S^a$. It's easy to establish a coordinate system compatible with $t$ and $s$, and consequently, the difference of a vector parallel-transported along this closed path will be
\begin{equation*}
\delta v^{a}={R_{cbd}}^{a}\cdot T^{c}dt\cdot S^{b}ds\cdot v^{d}
\end{equation*}
This equation \cite{GR} is intended to represent the effect of curvature. We have two curves and each one is of equal importance. Without loss of generality, we take the gauge
\begin{equation*}
\begin{split}
T^b\nabla_{\!b}S^a=0 \\
S^b\nabla_{\!b}T^a=0
\end{split}
\end{equation*}
Therefore, we construct the density potential as
\begin{equation}  \label{eq:original}
T^aS^b\nabla_{\!a}\!\nabla_{\!b}\phi =\varLambda \sqrt{R_{abcd}R^{abcd}}\tag{0}
\end{equation}
Within, $\varLambda$ is a constant. $\phi$ is the density potential, which is a scalar that can represent the difference of a parallel-transported vector in some degree. It's easily perceived that $\phi$ can reflect the difference of points as a result of curvature and we should remember the index $a$ and $b$ are limited by the curves. By analogy, it's natural to define the spacetime density $\rho$ by
\begin{equation}
\tag{\ref{eq:original}$'$}\label{eq:origprime}
\rho =-T^a\nabla_{\!a}\phi
\end{equation}
Solve the equation (0) and ($0'$), and then we'll have
\begin{equation}
\rho=-\varLambda \int \sqrt{R_{abcd}R^{abcd}}\,Tdt
\end{equation}
This is the computational equation of spacetime density.

Next we put forward two useful theorems.

\newtheorem{theorem}{Theorem}
\begin{theorem}
In the equation of spacetime density (1), if we use the coordinate $x^{\iota }$ as parameter, then the density equation can be simplified as
\begin{equation*}
\rho=-\varLambda \int \sqrt{R_{abcd}R^{abcd}}\,dx^{\iota }
\end{equation*}
In order to keep the balance of the indexes on both sides of the equation, some abstract index is left out.
\end{theorem}

\begin{proof}
There is an equation following for tangent vector and differential parameter
\begin{equation*}
\begin{split}
  T^{a}&=\frac{\partial x^{\iota }}{\partial t}X^{a}\\
  dt&=\frac{\partial t}{\partial x^{\iota }}dx^\iota
\end{split}
\end{equation*}

So we have
\begin{equation*}
Tdt=dx^{\iota }
\end{equation*}
Thus we gain theorem 1.
\end{proof}

\begin{theorem}
The coefficient of spacetime density equation satisfies
\begin{equation*}
\varLambda =\frac{c^{2}}{16\pi G^{2}M}
\end{equation*}
\end{theorem}

\begin{proof}
Because equation (1) always holds in any spacetime condition, if we want to solve $\varLambda$, we only need to solve its value in certain condition, and then it can represent the value of $\varLambda$ in all conditions. In General Relativity, the most simple spacetime condition is the static Schwarzschild Solution. So we need to solve $\varLambda$ in this condition.

For Schwarzschild Solution, we have the components of curvature tensor
\begin{widetext}
\begin{align*}
R_{0101}&=-\frac{2M}{r^{3}} & R_{0202}&=\frac{M}{r}\left ( 1\!-\!2M/r \right ) & R_{0303}&=\frac{M}{r}\left ( 1\!-\!2M/r \right )\sin ^{2}\!\theta\\
R_{2323}&=2Mr\sin \!^{2}\theta & R_{1212}&=-\frac{M}{r}\left ( 1\!-\!2M/r \right )^{-1} & R_{1313}&=-\frac{M}{r}\left ( 1\!-\!2M/r \right )^{-1}\sin ^{2}\!\theta
\end{align*}
And the metric
\begin{equation*}
\begin{split}
ds^2=&-\left ( 1\!-\!\frac{2M}{r} \right )dt^{2}+\left ( 1\!-\!\frac{2M}{r} \right )^{-1}\!dr^{2}+r^{2}\left ( d\theta ^{2}+\sin \!^{2} \theta \,d\varphi ^{2}\right )
\end{split}
\end{equation*}
\end{widetext}
Wherein the coordinate lines of $t$ and $r$ are the only two curves without self-intersection. We gain the spacetime density in this condition
\begin{equation*}
\rho =\varLambda \frac{2GM}{r^{2}}+C
\end{equation*}
Because the curvature tensor equals zero at infinity, so the spacetime density equals zero at infinity, too. Thus, the constant $C=0$.

Next, integrate the density and we have mass
\begin{equation*}
M=8\pi\! \varLambda G\!M\!R
\end{equation*}
And we also have the radius of the black hole
\begin{equation*}
R=\frac{2GM}{c^{2}}
\end{equation*}
Then we can solve the coefficient $\varLambda$.
\end{proof}

Using the theorem above, we can solve the spacetime formula in the condition of Schwarzschild Solution
\begin{equation*}
\rho =\frac{c^{2}}{8\pi Gr^{2}}
\end{equation*}
And find when $r\rightarrow 0$, there is $\rho\rightarrow \infty $, and then we can define this position as the position of matter when there is an observation. Also we discover that the spacetime density in this condition has nothing to do with any intrinsic property of matter but position.

\subsection{Spacetime Fluid}
In General Relativity, the electromagnetic field is described by energy-momentum tensor. This theory tries to describe the electromagnetic field as a kind of perfect fluid, which makes the theory able to be used microscopically.

In General Relativity, there is the energy-momentum tensor of electromagnetic field
\begin{equation*}
T_{ab}=\frac{1}{2\mu}\left ( F_{ac}{F_{b}}^{c}+^{\ast }\!\!{F_{ac}}^{\ast }\!{F_{b}}^{c} \right )
\end{equation*}
And the energy-momentum tensor of perfect fluid
\begin{equation*}
T_{ab}=\left ( \eta +p \right )U_{a}U_{b}+pg_{ab}
\end{equation*}

So if we consider the electromagnetic field as a kind of spacetime fluid, then the energy-momentum tensor of them must be equal. By analogy, we put forward the energy-momentum tensor of spacetime fluid
\begin{equation}
T_{ab}=\eta U_aU_b+ph_{ab}-\frac{1}{2}\left |c\eta\varepsilon_{abcd}V^cU^d\right |
\end{equation}
In the equation, there are two kinds of vector field
\begin{equation*}
\begin{split}
U_a&=Z_a\vec{l}+c\eta^{-\frac{1}{2}}\varepsilon^{\frac{1}{2}}F_{ab}Z^b\vec{m}\,-\eta^{-\frac{1}{2}}\mu^{-\frac{1}{2}} {}^ \ast\!F_{ac}Z^c\vec{n}\\
V_a&=Z_a\vec{l}+c\eta^{-\frac{1}{2}}\varepsilon^{\frac{1}{2}}F_{ab}Z^b\vec{n}\,+\eta^{-\frac{1}{2}}\mu^{-\frac{1}{2}} {}^ \ast\!F_{ac}Z^c\vec{m}
\end{split}
\end{equation*}
And two kinds of scalar field
\begin{equation*}
\begin{split}
&\eta= \frac{1}{2}\left (\varepsilon E^2+\mu^{-1}B^2 \right )\\
&p= -\frac{1}{2}\left (\varepsilon E^2+\mu^{-1}B^2 \right )
\end{split}
\end{equation*}
Wherein, $Z^a$ is the 4-velocity of any observer; $U^a$ is the 4-velocity field of spacetime fluid and $V^a$ is the symmetric of it; $\eta$ is the energy density and $p$ is the pressure; $\vec{l}$, $\vec{m}$ and $\vec{n}$ are orthonormalized quantity; $\varepsilon$ is vacuum dielectric constant and $\mu$ is permeability of vacuum; $E$ is electric field strength and $B$ is magnetic field strength; $h_{ab}$ is the induced metric and $\varepsilon_{abcd}$ is the adapted volume element; $c$ is light speed in vacuum.

Then, we can discuss the significance of each quantity in the spacetime fluid. First, it's easy to verify that the components of the energy-momentum tensor of spacetime fluid are the same as that of the electromagnetic field. We transform the electromagnetic field into a kind of spacetime fluid and construct the 4-velocity by electromagnetic field tensor. The orthonormalized quantity used in the 4-velocity of spacetime fluid has nothing to do with the geometric vector space. Its function is only to insure that (2) will always hold. Thus, the electromagnetic field is described by two vector fields and two scalar fields.

Next, we are going to study further on $\vec{l}$, $\vec{m}$ and $\vec{n}$ . We find that $U^a$ and $V^a$ have an interesting symmetry and it can be easily explained if $\vec{l}$, $\vec{m}$ and $\vec{n}$ have the form of the combination of wave function $e^{\mathbf{i}\theta}$ and spin matrix $\phi$ together with its conjugate $\bar{\phi}$. So make a transformation on $U_a$
\begin{equation}
|U_a\rangle\!=Z_a\phi^\mp \!+\!\eta^{-\frac{1}{2}}\mu^{-\frac{1}{2}}\left ( F_{ab}Z^be^{\mathbf{i}\theta}\phi^\pm\!\!-\!{}^ \ast\!F_{ac}Z^c\mathbf{i}e^{\mathbf{i}\theta}\phi^\pm\right )
\end{equation}
And see its conjugate
\begin{equation*}
\langle U_a|\!=Z_a\bar{\phi}^\mp \!+\!\eta^{-\frac{1}{2}}\mu^{-\frac{1}{2}}\left ( F_{ab}Z^be^{-\mathbf{i}\theta}\bar{\phi}^\pm\!\!+\!{}^ \ast\!F_{ac}Z^c\mathbf{i}e^{-\mathbf{i}\theta}\bar{\phi}^\pm\right )
\end{equation*}
Thus, we have
\begin{equation*}
\langle V_a|=\mathbf{i}\langle U_a|
\end{equation*}
So the energy-momentum tensor can be simplified as
\begin{equation}
T_{ab}=\eta \langle U_a|U_b\rangle+ph_{ab}-c\eta\left |{}^{\ast }\mathbf{i} \langle U_a|U_b\rangle\right |
\end{equation}

From (3) we discover that there is wave inside the electromagnetic field, which may be the origin of the electromagnetic wave. Equation (3) and Equation (4) are the final result of the discussion of spacetime fluid.
\subsection{Spacetime Torsion and Energy Exchange}
Einstein used the spacetime without the torsion when developing his General Relativity, which means he only considered the spacetime curvature but not contortion. Then Cartan established the spacetime with torsion and it became the Einstein-Cartan Theory. According to this theory, the distribution of matter will influence the torsion of spacetime. Hence, in the condition of this theory will there be some other phenomena.

In order to express these phenomena, first we consider the definition of torsion tensor
\begin{equation*}
\left ( \nabla _{\!\!a}\nabla_{\!b} -\nabla_{\!b} \nabla _{\!\!a} \right ) f=-{T^{c}}_{ab}\nabla _{\!\!c}f
\end{equation*}

For ordinary matter field $\mu$, we put $f =\mu$ into the equation, and consider the difference of a periodic transport along a closed path, then we gain
\begin{equation*}
\delta \mu=\int\!\!\!\!\! \int {T^{c}}_{ba}\nabla _{\!\!c}\,\mu \cdot I^{a}d\iota \cdot N^{b}d\nu
\end{equation*}
Hence, we see the energy of the field in connection with self-spin will dissipate. In this theory, this part of energy will be exchanged with the spacetime. The energy is released or taken in by the spacetime.

For further discussion, it's necessary to distinguish two kinds of different torsion tensor field: $_{o}{T^{c}}_{ab}$ and $_{i}{T^{c}}_{ab}$. $_{o}{T^{c}}_{ab}$ means the torsion tensor field produced by the matter except the target object, and $_{i}{T^{c}}_{ab}$ means the torsion tensor field produced by the target object itself. Due to the spin of the object, ${T^{c}}_{ab}$ will cause its energy to be absorbed by the spacetime, matter field losing energy, which is called positive energy exchange. And because of $_{i}{T^{c}}_{ab}$ is created by the object itself, it will make the spacetime energy release from spacetime to matter field, the energy of the matter field rise, which is called negative energy exchange.

To construct the equation of energy exchange, we need the theorems below.
\begin{theorem}
The length of space rotation angular velocity vector is invariant to any inertial observer or coordinate system.
\end{theorem}
\begin{proof}
In General Relativity, suppose there are world lines of a particle $G$, then for space rotation angular velocity vector is there an equation \cite{WG}
\begin{equation*}
g_{ab}\frac{Dw^{b}}{d\tau }=\varepsilon _{abcd}Z^{b}w^{c}\omega ^{d}
\end{equation*}
Within $D/d\tau$ indicates Fermi-Walker derivative, $\varepsilon _{abcd}$ is the volume element adapted with $g_{ab}$, $Z^{b}$ is the tangent vector of the world line, $w^{c}$ is the space vector in any point in any world line of $G$. Change all the indexes and then contract to get the length of $\omega ^{d}$
\begin{equation*}
n\frac{Dw^{a}}{d\tau }\frac{Dw_{a}}{d\tau }=\varepsilon^{2}Z^2w^{2}\omega ^{2}
\end{equation*}
Within $Z$, $w$, $\omega$, $\varepsilon$ is the length of $Z^b$, $w^{c}$, $\omega^d$, $\varepsilon _{abcd}$ with contraction. Also there is
\begin{equation*}
\frac{Dw^{a}}{d\tau }={h^{a}}_{b}Z^{c}\nabla _{c}w^{b}
\end{equation*}
Within ${h^{a}}_{b}$ is the projection map, and without loss of generality, suppose $w^c$ and $\omega^d$ is perpendicular, so we get
\begin{equation*}
\omega ^{2}=4w^{-2}Z^{c}\nabla _{c}w^{b}Z^{d}\nabla _{d}w_{b}
\end{equation*}
And discover all the quantities in the right side of equation are invariant to any inertial observer and coordinate system.
\end{proof}

\begin{theorem}
The angular momentum vector in General Relativity can be formulated by
\begin{equation*}
J^a=m\omega^a r^\iota r_\iota
\end{equation*}
Within, $m$ is the rest mass, $\omega^a$ is the space rotation angular velocity vector, $r^\iota$ is the displacement vector.
\end{theorem}
\begin{proof}
Define the displacement vector and angular momentum by
\begin{equation*}
\begin{split}
r_\iota &= \int \!\!\sqrt{{\delta_{\iota}}^{c}h_{c\iota}}\\
J_a&=\varepsilon _{abcd}Z^b P^c r^d
\end{split}
\end{equation*}
Wherein, $h_{c\iota}$ is the induced metric and $P^c$ is the 4-momentum.

Then, with the help of
\begin{equation*}
U_a=\varepsilon _{abcd}Z^br^c\omega^d
\end{equation*}
We can solve the formulation of the angular momentum vector.
\end{proof}

Then we can calculate the variation of energy. First consider positive energy exchange. For photons and other particles, there is quantized angular momentum
\begin{equation*}
J=\sqrt{j\left ( j+1 \right )}\hbar
\end{equation*}
According to the theorem, there is
\begin{equation*}
J=m\omega r^{2}
\end{equation*}
So we can solve them and gain
\begin{equation*}
\omega=\frac{\sqrt{j\left ( j+1 \right )}\hbar}{mr^{2}}
\end{equation*}
Hence, there is
\begin{equation*}
\delta t=T=\frac{2\pi mr^{2}}{\sqrt{j\left ( j+1 \right )}\hbar}
\end{equation*}
So we get the positive energy exchange rate
\begin{equation*}
\eta =\!\!\int \!\frac{\delta \mu }{\delta t}dr=\!\!\int dr\frac{\sqrt{j\left ( j+1 \right )}\hbar}{2\pi mr^{2}}\int\!\!\!\!\! \int {T^{c}}_{ba}\nabla _{\!\!c}\,\mu \cdot I^{a}d\iota \cdot N^{b}d\nu
\end{equation*}
Next consider the negative energy exchange of ordinary particle and black hole. For ordinary particle, it's obvious that there is
\begin{equation*}
\gamma=\int dr\frac{\sqrt{j\left ( j+1 \right )}\hbar}{2\pi mr^{2}}\int\!\!\!\!\! \int {T^{c}}_{ba}\nabla _{\!\!c}\,\mu \cdot I^{a}d\iota \cdot N^{b}d\nu
\end{equation*}
For black hole, because of
\begin{equation*}
J=Ma
\end{equation*}
So its negative energy exchange is
\begin{equation*}
\gamma=\int dr \frac{a}{2\pi r^{2}}\int\!\!\!\!\! \int {T^{c}}_{ba}\nabla _{\!\!c}\,\mu \cdot I^{a}d\iota \cdot N^{b}d\nu
\end{equation*}

Next we discuss the result separately. For photons, due to the nonexistence of rest mass, there is no negative energy exchange, but positive energy exchange. As a result, the energy of photons will be dissipated, whose rate depends on total torsion tensor field. For ordinary particles, when $_{o}{T^{c}}_{ab}$ does not appear or is so small that it can be ignored, there is ${T^{c}}_{ab} \approx $ $_{i}{T^{c}}_{ab}$. So the rate of two energy exchange is almost equal, resulting in the invariance of its mass. For the black hole, because there is no quantized angular momentum, there is no positive energy exchange. Supposing a KN black hole with spin, the mass will rise ceaselessly, whose rate depends on the torsion tensor field created by itself and its angular velocity.

Hence, this theory can explain Hubble's law and parts of the problems of dark matter. For Hubble's law, due to the loss of the energy of photons produced by stars in galaxies, the spacetime expands, resulting in the distance between galaxies increasing. If we suppose the galaxies homogeneously distribute, it's easy to see the spacetime expansion and distance is proportional. For some galaxies rotating too fast, it is possible that it is because the black hole absorbs the spacetime energy, the spacetime contracts, causing the galaxies difficult to break up. It's also possible for some photons released from some matter in the galaxy to run out of energy and become unable to be discovered, forming what is called ``dark matter''.

We need to declare that all quantities in connection with spacetime are the spacetime property in the theory, but there is a sort of special quantities called intrinsic spacetime property. Only some of the spacetime quantities such as spacetime curvature, torsion, density and spin belong to the intrinsic property.

\section{Quantization of General Relativity}
This theory finally will try to be combined with Quantum Mechanics. Quantum Mechanics describes phenomena microscopic and quantized, so observation and measurement become very important concepts. So, the theory needs to define observation and measurement.

\begin{Definition}
Signal is a kind of map $\chi\!\!:\!\!V^n\rightarrow V^n$, wherein $V^n$ is the Cartesian product of any vector space $V$ for $n$ times.
\end{Definition}
\begin{Definition}
Measurement is a set $\mathcal{F}$ of map $f$, satisfy\\
$\left (a  \right )$ $f\!:\!Q\rightarrow \mathbb{R}^n$ is continuous, wherein $Q$ is the set of certain physical quantities\\
$\left (b  \right )$ $\forall f,g\in \mathcal{F}$, and $f\!:\!q \mapsto x$, $g\!:\!q \mapsto y$, then $g\circ\! f^{-1}\!:\!x \mapsto y$ is smooth\\
$\left (c  \right )$ $\exists h$ such that when $f\!:\!q \mapsto x$, $f\!:\!p \mapsto y$ and $h\!:\!p \mapsto q$, $f\circ h\circ\! f^{-1}\!:\!x \mapsto y$ is smooth
\end{Definition}
\begin{Definition}
Observation is a signal, if\\
$\left (a  \right )$ It is non spacetime intrinsic property signal\\
$\left (b  \right )$ It is released by an object called signal emission source\\
$\left (c  \right )$ It can be measured, and when it becomes any observable quantity $\hat{\alpha }$, the result $a$ satisfy
\begin{equation*}
\hat{\alpha }\!\left | \varphi \right \rangle=a\left | \varphi \right \rangle
\end{equation*}
Within $\left | \varphi \right \rangle$ indicates some state of a system
\end{Definition}

From the definition can we see that signal changes the state of an object, which provides us with methods to test the state of a signal emission source by measuring the change of certain physical quantity of the particles in the instrument, which often appears as the interaction of quanta. The definition of measurement guarantees that we can use a series of real numbers to represent the state of the signal emission source, and satisfies the definition of $n$-dimensional differential manifold, which provides the symmetry of unit exchange. Observation is a special signal and the result of measurement is one of its eigenvalues.

According to the definition of observation and measurement above, there are two important points: \ding{192}There is observation without measurement. No matter whether we measure the object, it is the signal source, as long as it is releasing observation, with its wave function collapsed. \ding{193}The observation does not contain spacetime intrinsic property signal, which allow us to measure the spacetime intrinsic property without affecting the wave function in Quantum Theory.

Based on signal, measurement and observation, we try to combine the General Relativity with the Quantum Mechanics separately. First we have to introduce some new models and equations. It's necessary to leave out some property about vectors in some equation where geometrical vector space and Hilbert space both exist. Then, about the new model of particles is there the quantized matter distribution
\begin{equation*}
x\left ( r \right )=x\cdot \omega \left ( r \right )=x\cdot \left \langle \varphi\!\!\mid \!\!\varphi  \right \rangle
\end{equation*}
Wherein $\omega \left ( r \right )$ is the probability density, $x$ is any quantity in connection with the intrinsic property of a particle, putting a particle into a kind of perfect fluid. One of the most important examples is the mass distribution, after which we can construct the curved Quantum Equation.

Analogized from the wave equation of light, we gain the equation of free particles
\begin{equation*}
\hbar^2c^2\nabla_{\!\!a} \nabla^a |\varphi\rangle-m^2c^4|\varphi\rangle=0
\end{equation*}
And if we need to consider any gauge field $T^a$, we will have
\begin{equation}
\kappa^{2} \sigma^2 \hbar^2c^2\nabla_{\!\!a} \nabla^a |\varphi\rangle-\rho^2 m^2c^4|\varphi\rangle-\delta \kappa^{2}\sigma^2 T_aT^a |\varphi\rangle=0
\end{equation}
This is the basic curved Quantum Equation. Within, $\nabla_{\!\!a}$ is the curved derivative operator; $m$ is the rest mass of particle; $T^a$ is a possible potential field and $\delta$ is any number; $\kappa$, $\sigma$ and $\rho$ are matrices aimed at the spin.

By defining a new derivative operator as
\begin{equation*}
\hat\nabla_{\!\!a} =\nabla_{\!\!a}+\!T_{a}
\end{equation*}
And with the help of a gauge condition
\begin{equation*}
\nabla_{\!\!a}T^{a}=0
\end{equation*}
And combining the matrices together, we can rewrite equation (5) in natural unit as
\begin{equation*}
\sigma^2 \hat{\nabla}_{\!\!a} \hat{\nabla}^{a} |\varphi\rangle-\rho^2 m^2|\varphi\rangle=0
\end{equation*}

Define the 4-probability-density-current as
\begin{equation}
J^a=\mathbf{i}\hbar^2c^2\left (\langle \varphi|\nabla^a|\varphi \rangle- |\varphi\rangle\nabla^a\langle\varphi| \right )
\end{equation}
Then it's easy to see (6) leads to the conservation equation below
\begin{equation*}
\nabla_{\!\!a}J^a=0
\end{equation*}

Then we need the hypotheses below
\begin{Hypothesis}
Arbitrary particle can be expressed by quantized matter distribution.
\end{Hypothesis}
\begin{Hypothesis}
Evolution of particle states satisfy the curved Quantum Equation.
\end{Hypothesis}
\begin{Hypothesis}
Interaction between particles are described by spacetime property.
\end{Hypothesis}

Next we illustrate the rationality of these hypotheses: The General Relativity and Quantum Mechanics should be treated on the merits of each case. When an object has continuous observation, General Relativity functions in one of its neighborhood with continuous observation, and all the contents in Quantum Mechanics embeds into the General Relativity with the wave function collapsed to $\delta$ function, for no matter the observation is created by the object itself or reflected from the instrument, its wave function should collapse in one of its neighborhood with continuous observation. Owing to the non-measurement observation, we can deem that the object has certain position and certain coordinate time and can look upon it as spacetime point in General Relativity without considering the nondeterministic finite automation. When the object has no observation or has only discontinuous observation, Quantum Mechanics functions, and we can construct the spacetime background by quantized matter distribution. Thus, the simultaneous formulas will give all the information of a particle.

Hence, the connection between matter distribution and wave function is established. In the critical point, the variance of wave function indicates the variance of matter distribution, and then affects the spacetime curvature. For instance, suppose there is an electron in a hydrogen atom in ground state with the wave function decreased exponentially. In this condition, the quantized mass distribution functions, and the effect of electrons on spacetime is described by quantized mass distribution. However, if some time a photon comes and gets the electron to have continuous observation, then the wave function of the electron will collapse immediately. The electron becomes a point particle at this time, and the effect on spacetime changes to the effect of point particle mass distribution on spacetime.

Using the conclusion above, we can try to give the probability explanation for Quantum Equation a new meaning. First of all, we suppose there exists a possible energy exchange in the intersection of particles, which is formulated by
\begin{equation*}
\delta H=\sum_{i,j}^{n}c^{2}\!\!\int \min\left \{ m_i\left \langle \varphi_i|\varphi_i \right \rangle\!,m_j\left \langle \varphi_j|\varphi_j \right \rangle \right \}dv
\end{equation*}
And if $T_i$ is the energy of particle $i$ that may be carried by itself or gain from outside, then the possible energy exchange will be
\begin{equation}
\delta E=\sum_{i,j}^{n}c^{2}\!\!\int \min\left \{ m_i\left \langle \varphi_i|\varphi_i \right \rangle\!,m_j\left \langle \varphi_j|\varphi_j \right \rangle \right \}dv\!+\!T_i\!+\!T_j
\end{equation}
If this quantity at point $p$ equals (or greater than for continuous condition) the energy difference of any two energy levels of one of the particles $i$, then the energy $\delta H$ together with $T_i$ and $T_j$ will be transferred to particle $i$ and at the same time, the wave function of particle $j$ collapses, taking all its energy to point $p$ and releasing the energy equal to $T_i+T_j$.

From (7) we discover that the possibility of finding a particle at a certain point is influenced by the wave function and the possible energy $T_i$. This can be used to explain the point or line that we see in the instrument. If the energy of the instrument particles (for example, the kinetic energy of atoms) is various enough, the possibility of finding a target particle at a certain place will be proportional to the wave function, and hence form the certain image (for example, the interference fringes).

Thus, the probability explanation of Quantum Equation will only hold in the condition that the environment will give various enough $T_i$ for the particles. But this condition is always satisfied since our instrument usually works at the temperature of more than $270K$, which provides at least enough kinetic energy of atoms.

\section{Discussion and Final Remarks}
This article advances the GUSFT Theory and aims at establishing a framework of united theory. We mainly present certain spacetime fields and the application of matter distribution in quantization of General Relativity.

The spacetime density provides us with an inspiration that there is mass hidden inside the spacetime. Together with the spacetime torsion, we find that this configuration of mass plays an important part in the evolution of universe.

The spacetime fluid may reflect to the essence of electromagnetic field. We compare the energy-momentum tensor of electromagnetic field and perfect fluid, from which the 4-velocity of electromagnetic field can be formulated. Using the 4-velocity, we rewrite the energy-momentum tensor and discover what may be the origin of electromagnetic wave.

The new approach aimed at the quantization of General Relativity begins with the definition of signal, measurement and observation. The quantized matter distribution is defined as the description of matter in a system so that the theory can overcome the difficulty raised by the uncertainty of matter displacement in Quantum Theory. And thus, the evolution is provided by the simultaneous equations
\begin{equation*}
\begin{split}
\sigma^2 \hat{\nabla}_{\!\!a} & \hat{\nabla}^{a} |\varphi\rangle-\rho^2 m^2|\varphi\rangle=0 \\
R_{ab} & -\frac{1}{2}Rg_{ab}=\kappa T_{ab}
\end{split}
\end{equation*}

By the discussion above, we are able to point out that the possibility explanation in Quantum Theory is only a statistical effect. The energy of instrument particle may lead to some bizarre phenomena in Quantum Theory.

In summary, this article constructs the GUSFT Theory and combines the General Relativity and Quantum Mechanics to some extent by quantized matter distribution, which provides another inspiration of united theory.

\end{document}